\documentclass{article}[11pt]
\usepackage{graphicx,epsfig,color}

\usepackage{latexsym}

\newtheorem{theorem}{Theorem}
\newtheorem{lemma}[theorem]{Lemma}

\newtheorem{definition}[theorem]{Definition}

\newcommand{\scrod}{\quad\nopagebreak}

\newenvironment{proof}
{\bigskip\noindent\textbf{Proof~}} {\marginpar{$\Box$}\bigskip}

\begin{document}

\date{}

\title{Algorithms for Testing Monomials in Multivariate Polynomials}

\author{Zhixiang Chen \hspace{5mm} Bin Fu \hspace{5mm} Yang Liu
\hspace{5mm} Robert Schweller
 \\ \\
Department of Computer Science\\
 University of Texas-Pan American\\
 Edinburg, TX 78539, USA\\
\{chen, binfu, yliu, schwellerr\}@cs.panam.edu\\\\
} \maketitle

\begin{abstract}
This paper is our second step towards developing a theory of
testing monomials in multivariate polynomials. The central
question is to ask whether a polynomial represented by an
arithmetic circuit has some types of monomials in its sum-product
expansion. The complexity aspects of this problem and its variants
have been investigated in our first paper by Chen and Fu (2010),
laying a foundation for further study. In this paper, we present
two pairs of algorithms. First, we prove that there is a
randomized $O^*(p^k)$ time algorithm for testing $p$-monomials in
an $n$-variate polynomial of degree $k$ represented by an
arithmetic circuit, while a deterministic $O^*(6.4^k + p^k)$ time
algorithm is devised when the circuit is a formula, here $p$ is a
given prime number. Second, we present a deterministic $O^*(2^k)$
time algorithm for testing multilinear monomials in
$\Pi_m\Sigma_2\Pi_t\times \Pi_k\Pi_3$ polynomials, while a
randomized $O^*(1.5^k)$ algorithm is given for these polynomials.
The first algorithm extends the recent work by Koutis (2008) and
Williams (2009) on testing multilinear monomials.  Group algebra
is exploited in the algorithm designs, in corporation with the
randomized polynomial identity testing over a finite field by
Agrawal and Biswas (2003), the deterministic noncommunicative
polynomial identity testing by Raz and Shpilka (2005) and the
perfect hashing functions by Chen {\em at el.} (2007). Finally, we
prove that testing some special types of multilinear monomial is
W[1]-hard, giving evidence that testing for specific monomials is
not fixed-parameter tractable.

\end{abstract}

\section{Introduction}
\subsection{Overview}
We begin with two examples to exhibit the motivation and necessity
of the study about the monomial testing problem for multivariate
polynomials. The first is about testing a  $k$-path in any given
undirected graph $G=(V,E)$ with $|V| = n$, and the second is about
the satisfiability problem. Throughout this paper, polynomials
refer to those with multiple variables.

For any fixed integer $c\ge 1$, for each vertex $v_i \in V$,
define a polynomial $p_{k,i}$ as follows:
\begin{eqnarray}
p_{1,i}  &=&  x_i^c, \nonumber \\
p_{k+1,i} &=&  x_i^c   \left(\sum_{(v_i,v_j)\in E} p_{k,j}\right),
\ k
>1. \nonumber
\end{eqnarray}
We define a polynomial for $G$ as
\begin{eqnarray}
p(G, k)  &=&  \sum^{n}_{i=1} p_{k,i}. \nonumber
\end{eqnarray}
Obviously, $p(G,k)$ can be represented by an arithmetic circuit.
It is easy to see that the graph $G$ has a $k$-path $v_{i_1}\cdots
v_{i_k}$ iff $p(G, k)$ has a monomial $x_{i_1}^c\cdots
x_{i_k}^c$ of degree $ck$ in its sum-product expansion. $G$ has a
Hamiltonian path iff $p(G, n)$ has the monomial $x_1^c\cdots
x_n^c$ of degree $cn$ in its sum-product expansion. One can also
see that a path with some loop can be characterized by a monomial
as well. Those observations show that testing monomials in
polynomials is closely related to solving $k$-path, Hamiltonian
path and other problems about graphs. When $c=1$, $x_{i_1}\cdots
x_{i_k}$ is multilinear. The problem of testing multilinear
monomials has recently been exploited by Koutis \cite{koutis08}
and Williams \cite{williams09} to design innovative randomized
parameterized algorithms for the $k$-path problem.

Now, consider any CNF formula $f= f_1 \wedge \cdots \wedge f_m$, a
conjunction of $m$ clauses with each clause $f_i$ being a
disjunction of some variables or negated ones. We may view
conjunction as multiplication and disjunction as addition, so $f$
looks like a {\em "polynomial"}, denoted by $p(f)$. $p(f)$ has a
much simpler $\Pi\Sigma$ representation, as will be defined in the
next section, than general arithmetic circuits. Each {\em
"monomial"} $\pi = \pi_1 \ldots \pi_m$ in the sum-product
expansion of $p(f)$ has a literal $\pi_i$ from the clause $f_i$.
Notice that  a boolean variable $x \in Z_2$ has two properties of
$x^2 = x$ and $x \bar{x} = 0$. If we could realize these
properties for $p(f)$ without unfolding it into its sum-product,
then $p(f)$ would be a {\em "real polynomial"} with two
characteristics: (1) If $f$ is satisfiable then $p(f)$ has a
multilinear monomial, and (2) if $f$ is not satisfiable then
$p(f)$ is identical to zero. These would give us two approaches
towards testing the satisfiability of $f$. The first is to test
multilinear monomials in $p(f)$, while the second is to test the
zero identity of $p(f)$. However, the task of realizing these two
properties with some algebra to help transform $f$ into a needed
polynomial $p(f)$ seems, if not impossible, not easy. Techniques
like arithmetization in Shamir \cite{shamir92} may not be suitable
in this situation. In many cases, we would like to move from $Z_2$
to some larger algebra so that we can enjoy more freedom to use
techniques that may not be available when the domain is too
constrained. The algebraic approach within $Z_2[Z^k_2]$ in Koutis
\cite{koutis08} and Williams \cite{williams09} is one example
along the above line. It was proved in Bshouty {\em et al.}
\cite{bshouty95}  that extensions of  DNF formulas over $Z^n_2$ to
$Z_N$-DNF formulas over the ring $Z^n_N$ are learnable by a
randomized algorithm with equivalence queries, when $N$ is large
enough. This is possible because a larger domain may allow more
room to utilize randomization.

There has been a long history in theoretical computer science with
heavy involvement of studies and applications of polynomials. Most
notably, low degree polynomial testing/representing and polynomial
identity testing have played invaluable roles in many major
breakthroughs in complexity theory. For example, low degree
polynomial testing is involved in the proof of the PCP Theorem,
the cornerstone of the theory of computational hardness of
approximation and the culmination of a long line of research on IP
and PCP (see, Arora {\em at el.} \cite{arora98} and Feige {\em et
al.} \cite{feige96}). Polynomial identity testing has been
extensively studied due to its role in various aspects of
theoretical computer science (see, for examples, Chen and Kao
\cite{chen00}, Kabanets and Impagliazzo \cite{kabanets03}) and its
applications in various fundamental results such as Shamir's
IP=PSPACE \cite{shamir92} and the AKS Primality Testing
\cite{aks04}. Low degree polynomial representing
\cite{minsky-papert68} has been sought for so as to prove
important results in circuit complexity, complexity class
separation and subexponential time learning of boolean functions
(see, for examples, Beigel \cite{beigel93}, Fu\cite{fu92}, and
Klivans and Servedio \cite{klivans01}). These are just a few
examples. A survey of the related literature is certainly beyond
the scope of this paper.

The above two examples of the $k$-path testing and satisfiability
problems, the rich literature about polynomial testing and many
other observations have motivated us to develop a new theory of
testing monomials in polynomials represented by arithmetic
circuits or even simpler structures. The monomial testing problem
is related to, and somehow complements with, the low degree
testing and the identity testing of polynomials. We want to
investigate various complexity aspects of the monomial testing
problem and its variants with two folds of objectives. One is to
understand how this problem relates to critical problems in
complexity, and if so to what extent. The other is to exploit
possibilities of applying algebraic properties of polynomials to
the study of those critical problems. As a first step, Chen and Fu
\cite{chen-fu10} have proved a series of results: The multilinear
monomial testing problem for $\Pi\Sigma\Pi$ polynomials is
NP-hard, even when each clause has at most three terms. The
testing problem for $\Pi\Sigma$ polynomials is in P, and so is the
testing for two-term $\Pi\Sigma\Pi$ polynomials. However, the
testing for a product of one two-term $\Pi\Sigma\Pi$ polynomial
and another $\Pi\Sigma$ polynomial is NP-hard. This type of
polynomial products is, more or less, related to the polynomial
factorization problem. We have also proved that testing
$c$-monomials for two-term $\Pi\Sigma\Pi$ polynomials is NP-hard
for any $c > 2$, but the same testing is in P for $\Pi\Sigma$
polynomials. Finally, two parameterized algorithms have been devised
for three-term $\Pi\Sigma\Pi$ polynomials and products of two-term
$\Pi\Sigma\Pi$ and $\Pi\Sigma$ polynomials. These results have
laid a basis for further study about testing monomials.

\subsection{Contributions and Methods}

The major contributions of this paper are two pairs of algorithms.
For the first pair, we prove that there is a randomized $O^*(p^k)$
time algorithm for testing $p$-monomials in an $n$-variate
polynomial of degree $k$ represented by an arithmetic circuit,
while a deterministic $O^*(6.4^k + p^k)$ time algorithm is devised
when the circuit is a formula, here $p$ is a given prime number.
The first algorithm extends two recent algorithms for testing
multilinear monomials, the $O^*(2^{3k/2})$ algorithm by Koutis
\cite{koutis08} and the $O(2^k)$ algorithm by Williams
\cite{williams09}. Koutis \cite{koutis08} initiated the
application of group algebra $Z_2[Z^k_2]$ to randomized testing of
multilinear monomials in a polynomial. Williams \cite{williams09}
incorporated the randomized Schwartz-Zippel polynomial identity
testing with the group algebra $\mbox{GF}(2^{\ell})[Z^k_p]$ for
some relatively small ${\ell}$ in comparison with $k$ to achieve
the design of his algorithm. The success of applying group algebra
to designing multilinear monomial testing algorithms is based on
two simple but elegant properties found by Koutis, by which
annihilating non-multilinear monomials is possible via
replacements of variables by vectors in $Z^k_2$. When extending
the group algebra from $Z_2[Z^k_p]$ to $Z_p[Z^k_p]$ for a given
prime $p$ these two properties, as addressed in Section
\ref{sec-2}, are fortunately no longer valid. To make the matter
worse, the Schwartz-Zippel algorithm is not applicable to the
larger algebra due to the lack of these two properties. Nevertheless,
we find new characteristics about $Z_p[Z^k_p]$ and integrate these
with a more powerful randomized polynomial identity testing
algorithm by Agrawal and Biswas \cite{agrawal-biswas03} to
accomplish the design of our algorithm. Our deterministic
algorithm is obtained via derandomizing the two random processes involved
in the first algorithm:
deterministic selection of a set of linearly independent vectors for an unknown
monomial to guarantee its survivability from vector replacements;
and deterministic polynomial identity testing. The first part is
realized with the perfect hashing functions by Chen {\em at el.}
\cite{jianer-chen07}, while the second is carried out by the Raz
and Shpilka \cite{raz05} algorithm for noncommunicative
polynomials.

For the second pair of our algorithms, we present a deterministic
$O^*(2^k)$ time algorithm for testing multilinear monomials in
$\Pi_m\Sigma_2\Pi_t\times \Pi_k\Pi_3$ polynomials, while a
randomized $O^*(1.5^k)$ algorithm is given for these polynomials.
It has been proved in Chen and Fu \cite{chen-fu10} that testing
multilinear monomials in $\Pi_m\Sigma_2\Pi_t$ or $\Pi_k\Pi_3$
polynomials is solvable in polynomial time. However, the problem
becomes NP-hard for $\Pi_m\Sigma_2\Pi_t\times \Pi_k\Pi_3$
polynomials. Our two algorithms use the quadratic algorithm by
Chen and Fu \cite{chen-fu10} for testing multilinear monomials in
$\Pi_m\Sigma_2\Pi_t$ polynomials as the base case algorithm. Both
new algorithms improve the $O^*(3^k)$ algorithm in \cite{chen-fu10}.

Finally, we prove that testing some special types of multilinear
monomials is W[1]-hard, giving evidence that testing for specific monomials  is
not fixed-parameter tractable. One shall notice that difference between the general
monomial testing and the specific monomial testing. The former asks for the existence of
{\em "any one''} from a set of possibly  many monomials that are needed. The latter asks for
{\em "a specific one''} from the set.

\subsection{Organization}
The rest of the paper is organized as follows. In Section 2, we
introduce the necessary notations and definitions. In Section 3,
we prove new properties about the group algebra $Z_p[Z^k_p]$ to help
annihilate any monomials that are not $p$-monomials. These properties are then
integrated with the randomized polynomial identity testing over a
finite field to help design the randomize $p$-monomial testing
algorithm. In Section 4, the two randomized processes involved in
the randomized algorithm obtained in the previous section will be
derandomized for polynomials represented by formulas. The success
is based on combining deterministic construction of perfect
hashing functions with deterministic noncommunicative polynomial
identity testing. Section 5 first presents a deterministic
parameterized algorithm for testing multilinear monomials in
$\Pi_m\Sigma_2\Pi_t \times \Pi_k\Sigma_3$ polynomials, and then
gives a more efficient randomized parameterized algorithm for the
these polynomials. Finally, we show in Section 5 that testing some
special type of multilinear monomials, called $k$-clique
monomials, is W[1]-hard.

\section{Preliminaries}
\subsection{Notations and Definitions}

For variables $x_1, \dots, x_n$, let ${\cal P} [x_1,\cdots,x_n]$
denote the communicative ring of all the $n$-variate polynomials
with coefficients from a finite field ${\cal P}$. For $1\le i_1 <
\cdots <i_k \le n$, $\pi =x_{i_1}^{j_1}\cdots x_{i_k}^{j_k}$ is
called a monomial. The degree of $\pi$, denoted by
$\mbox{deg}(\pi)$, is $\sum^k_{s=1}j_s$. $\pi$ is multilinear, if
$j_1 = \cdots = j_k = 1$, i.e., $\pi$ is linear in all its
variables $x_{i_1}, \dots, x_{j_k}$. For any given integer $c\ge
1$, $\pi$ is called a $c$-monomial, if $1\le j_1, \dots, j_k < c$.

An arithmetic circuit, or circuit for short, is a direct acyclic
graph with $+$ gates of unbounded fan-in, $\times$ gates of fan-in
two, and all terminals corresponding to variables. The size,
denoted by $s(n)$, of a circuit with $n$ variables is the number
of gates in it. A circuit is called a {\rm formula}, if the
fan-out of every gate is at most one, i.e., its underlying direct
acyclic graph is a tree.

By definition, any polynomial $F(x_1,\dots,x_n)$ can be expressed
as a sum of a list of monomials, called the sum-product expansion.
The degree of the polynomial is the largest degree of its
monomials in the expansion. With this expression, it is trivial to
see whether $F(x_1,\dots,x_n)$ has a multilinear monomial, or a
monomial with any given pattern. Unfortunately, this expression is
essentially problematic and infeasible to realize, because a
polynomial may often have exponentially many monomials in its
expansion.

In general, a polynomial $F(x_1,\dots,x_n)$ can be represented by
a circuit or some even simpler structure as defined in the
following. This type of representation is simple and compact and
may have a substantially smaller size, say, polynomially in $n$,
in comparison with the number of all monomials in the sum-product
expansion. The challenge is how to test whether $F(x_1,\dots,x_n)$
has a multilinear monomial, or some other needed monomial, efficiently
without unfolding it into its sum-product expansion?

Throughout this paper, the $O^*(\cdot)$ notation is used to
suppress $\mbox{poly}(n,k)$ factors in time complexity bounds.

\begin{definition}\scrod
Let $F(x_1,\dots,x_n)\in {\cal P}[x_1,\dots,x_n]$ be any given
polynomial. Let $m, s, t\ge 1$ be integers.
\begin{itemize}
\item $F(x_1,\ldots, x_n)$ is said to be a $\Pi_m\Sigma_s\Pi_t$
polynomial, if $F(x_1,\dots,x_n)=\prod_{i=1}^t F_i$, $F_i =
\sum_{j=1}^{r_i} X_{ij}$ and $1\le r_i \le s$, and $X_{ij}$ is a product of variables with
$\mbox{deg}(X_{ij})\le t$. We call each $F_i$ a clause. Note that
$X_{ij}$ is not a monomial in the sum-product expansion of
$p(x_1,\dots,x_n)$ unless $m=1$. To differentiate this subtlety,
we call $X_{ij}$ a term.

\item In particular, we say $F(x_1,\dots,x_n)$ is a
$\Pi_{m}\Sigma_s$ polynomial, if it is a
 $\Pi_m\Sigma_s\Pi_1$ polynomial. Here, each clause in $f$ is a linear addition
 of single variables. In other word, each term has degree $1$.


\item $F(x_1,\dots,x_n)$ is called a $\Pi_m\Sigma_s\Pi_t \times
\Pi_k\Sigma_{\ell}$ polynomial, if $F(x_1,\dots,x_n) = f_1 \cdot f_2$
such that $f_1$ is a $\Pi_m\Sigma_s\Pi_t$ polynomial and $f_2$ is
a $\Pi_k\Sigma_{\ell}$ polynomial.

\end{itemize}
\end{definition}

When no confusion arises from the context, we use $\Pi\Sigma\Pi$
and $\Pi\Sigma$ to stand for $\Pi_m\Sigma_s\Pi_t$ and
$\Pi_m\Sigma_s$, respectively.




\subsection{The Group Algebra $F[Z_p^k$]}
For any prime $p$ and integer $k \ge 2$, we consider the group
$Z^k_p$ with the multiplication $\cdot$ defined as follows. For
$k$-dimensional column vectors $\vec{x}, \vec{y} \in Z^k_p$ with
$\vec{x} = (x_1, \ldots, x_k)^T$ and $\vec{y} = (y_1, \ldots,
y_k)^T$,
\begin{eqnarray}
\vec{x} \cdot \vec{y} &=& (x_1+y_1 \pmod{p}, \ldots, x_k+y_k
\pmod{p}).
\end{eqnarray}
$\vec{\bf 0}=(0, \ldots, 0)^T$ is the zero element in the group.
For any field $F$, the group algebra $F[Z^k_p]$ is defined as
follows. Every element $u \in F[Z^k_p]$ is a linear addition of the form
\begin{eqnarray}\label{exp-2}
 u &=& \sum_{\vec{x}\in Z^k_p, a_{\vec{x}}\in F} a_{\vec{x}} \vec{x}.
\end{eqnarray}
For any element
\begin{eqnarray}
v &=& \sum_{\vec{x}\in Z^k_p, b_{\vec{x}}\in F} b_{\vec{x}}
\vec{x}, \nonumber
\end{eqnarray}
We define
\begin{eqnarray}
 u + v  &=& \sum_{a_{\vec{x}, b_{\vec{x}}}\in F,\  \vec{x}\in Z^k_p}  (a_{\vec{x}}+b_{\vec{x}}\pmod{p})
 \vec{x}, \  \mbox{and} \\
u \cdot v &=& \sum_{a_{\vec{x}}, b_{\vec{y}}\in F, \mbox{ and }
\vec{x}, \vec{y}\in Z^k_p}  (a_{\vec{x}} b_{\vec{y}}\pmod{p})
(\vec{x}\cdot \vec{y}).
\end{eqnarray}
For any scalar $w \in F$,
\begin{eqnarray}
 w u &=& a \left(\sum_{\vec{x}\in Z^k_p, \ a_{\vec{x}}\in F} a_{\vec{x}} \vec{x}\right)
 = \sum_{\vec{x}\in Z^k_p,\  a_{\vec{x}}\in F} (w a_{\vec{x}} \pmod{p})\vec{x}.
\end{eqnarray}
The zero element in $F[Z^k_p]$ is the one  as represented in expression
(\ref{exp-2}) with zero coefficients in $F$:
\begin{eqnarray}
{\bf  0} &=& \sum_{\vec{x}\in Z^k_p} 0 \vec{x} = 0\vec{\bf 0}.
\end{eqnarray}
The identity element in $F[Z^k_p]$ is
\begin{eqnarray}
 {\bf 1} &=& 1 \vec{\bf 0} = \vec{\bf 0}.
\end{eqnarray}

For any vector $\vec{v} =(v_1, \ldots, v_k)^T \in Z_p^k$, for
$i\ge 0$, let
\begin{eqnarray}
(\vec{v})^i &=& (i v_1 \pmod{p}, \ldots, i v_k\pmod{p})^T.
\nonumber
\end{eqnarray}
 In particular, we have
 \begin{eqnarray}
 & & (\vec{v})^0 = (\vec{v})^p =
\vec{\bf 0}. \nonumber
\end{eqnarray}

When it is clear from the context, we will simply use $x y $ and
$x+y$ to stand for  $x y(\bmod{p})$ and $x+y \pmod{p}$,
respectively.

\section{Randomized Testing of $p$-Monomials}\label{sec-2}

Group algebra $Z_2[Z^k_2]$ was first used by Koutis
\cite{koutis08} and later by Williams \cite{williams09} to devise
a randomized $O^*(2^k)$ time  algorithm to test multilinear
monomials in $n$-variate polynomials represented by arithmetic
circuits. We shall extend $Z_2[Z^k_2]$ to
$Z_p[Z^d_p]$ to test $p$-monomials for some $d>k$.  Two key
properties in $Z_2[Z^k_2]$, as first found by Koutis
\cite{koutis08}, that are crucial to multilinear monomial testing
are unfortunately no longer valid in $Z_p[Z^d_p]$. Instead, we
establish new properties in Lemmas \ref{lem3} and \ref{lem4}.
Also, the Schwartz-Zippel algorithm \cite{motwani95} for
randomized polynomial identity testing adopted by Williams
\cite{williams09} is not applicable to our case. Instead, we have
to use a more advanced randomized polynomial identity testing
algorithm, the Agrawal and Biswas algorithm
\cite{agrawal-biswas03}.

Let $p$ be a prime number. Following conventional notations in
linear algebra, for any vectors $\vec{v}_1, \ldots, \vec{v}_t \in
Z^k_p$ with $k\ge 1$ and $t\ge 1$, let $\mbox{span}(\vec{v}_1,
\ldots, \vec{v}_t)$ be the linear space spanned by these vectors.
That is,
\begin{eqnarray}
\mbox{span}(\vec{v}_1, \ldots, \vec{v}_t) &=&
\{a_1\vec{v}_1+\cdots+a_t\vec{v}_t | a_1, \ldots, a_t \in Z_p\}.
\nonumber
\end{eqnarray}

We first give two simple properties about $\pmod{p}$ operation.

\begin{lemma}\label{lem1}
For any $x, y \in Z_p$, we have $(x + y)^p = x^p + y^p \pmod{p}$.
\end{lemma}

\begin{proof} $(x+y)^p = \sum^p_{i=0} (^p_i) x^{p-i} y^i
     =  x^p + y^p + \sum^{p-1}_{i=1} (^p_i)x^{p-i}y^{i}. $
Since $p$ is prime, $(^p_i)$ has a factor $p$, implying $(^p_i) =
0 \pmod{p}$, $1\le i\le p-1$. Hence, $(x+y)^p = x^p + y^p
\pmod{p}$.
\end{proof}

\begin{lemma}\label{lem2}
For any $x, y \in Z_p$, we have $((p-1)x + y)^p = (p-1)x^p + y^p
\pmod{p}$.
\end{lemma}

\begin{proof}
By Lemma \ref{lem1}, $((p-1)x + y)^p \equiv (p-1)^px^p + y^p
\pmod{p}$. By Fermat's Little Theorem, $(p-1)^p = (p-1) \pmod{p}$.
Thus, $((p-1)x + y)^p = (p-1)x^p + y^p \pmod{p}$.
\end{proof}

The first crucial, though simple, property observed by Koustis
\cite{koutis08} about testing multilinear monomials is that
replacing any variable $x$ by $(\vec{v}+\vec{\bf 0})$ will
annihilate $x^t$ for any $t\ge 2$, where $\vec{v}\in Z_2^k$ and
$\vec{v}_0$ is the zero vector. This property is not valid in
$Z_p[Z^d_p]$. However, we shall prove the following lemma that
helps annihilate any monomials that are not $p$-monomials.

\begin{lemma}\label{lem3}
Let $\vec{v}_0 \in Z_p^d$  be the zero vector and $\vec{v}_i \in
Z_p^d$ be any vector. Then, we have
\begin{eqnarray}\label{exp-lem3}
((p-1)\vec{v}_i + \vec{v}_0)^p &=& {\bf 0},
\end{eqnarray}
i.e., the zero element in $Z_p[Z_p^d]$.
\end{lemma}

\begin{proof}
By Lemma \ref{lem2}, we have $((p-1)\vec{v}_i + \vec{v}_0)^p =
(p-1)(\vec{v}_i)^p + (\vec{v}_0)^p \pmod{p}
  =  (p-1)\vec{v}_0 + \vec{v}_0 = p\vec{v}_0 \pmod{p} = {\bf 0}.$
\end{proof}

The second crucial property found by Koutis \cite{koutis08} has
two parts: (a) Replacing variables $x_{i_j}$ in a multilinear
monomial $x_{i_1}\cdots x_{i_k}$ with $(\vec{v}_{i_j}+\vec{v}_0)$
will annihilate the monomial, if  the vectors $\vec{v}_{i_j}$ are
linearly dependent in $Z_2^k$. (b) If these vectors are linearly
independent, then the sum-product expansion of the monomial after
the replacements will yield a sum of all $2^k$ vectors in $Z^k_2$.
However, neither (a) nor (b) is in general true in $Z_p[Z^k_p]$.
Fortunately, we have the following lemma, though not as {\em
"structurally"} perfect as (b).

\begin{lemma}\label{lem4}
Let $x_1^{m_1}\cdots x_{t}^{m_t}$ be any given $p$-monomial of degree
$k$.  If vectors $\vec{v}_1, \ldots, \vec{v}_t \in Z_p^d$ are
linearly independent, then there are nonzero coefficients $c_i \in
Z_p$ and distinct vector $\vec{u}_j \in Z_p^d$ such that
\begin{eqnarray}\label{exp-lem4}
((p-1)\vec{v}_1 + \vec{v}_0)^{m_1} \cdots ((p-1)\vec{v}_t +
\vec{v}_0)^{m_t} &=& c_1\vec{\bf 0} + \left(
\sum^{(m_1+1)(m_2+1)\cdots (m_t+1)}_{i=2}\right) c_i \vec{u}_i,
\end{eqnarray}
where $c_1 = 1$.
\end{lemma}

\begin{proof}
\begin{eqnarray}
& & ((p-1)\vec{v}_1 + \vec{v}_0)^{m_1} \cdots ((p-1)\vec{v}_t +
\vec{v}_0)^{m_t} \nonumber \\
& & = \left(\sum^{m_1}_{i_1=0}(^{m_1}_{i_1})(p-1)^{i_1}
(\vec{v}_1)^{i_1}\right)
\left(\sum^{m_2}_{i_2=0}(^{m_2}_{i_2})(p-1)^{i_2}
(\vec{v}_2)^{i_2}\right) \cdots
\left(\sum^{m_t}_{i_t=0}(^{m_t}_{i_t})(p-1)^{i_t}
(\vec{v}_t)^{i_t}\right)  \nonumber \\
& & =  \sum^{m_1}_{i_1=0} \sum^{m_2}_{i_2=0} \cdots
\sum^{m_t}_{i_t=0} (^{m_1}_{i_1})(^{m_2}_{i_2}) \cdots
(^{m_t}_{i_t}) (p-1)^{i_1+t_2\cdots+i_t} (\vec{v}_1)^{i_1}
(\vec{v}_2)^{i_2} \cdots (\vec{v}_t)^{i_t}
\end{eqnarray}
As noted in the previous section, in the vector space $Z_p^d$, we have
\begin{eqnarray}\label{exp-a}
(\vec{v}_1)^{i_1} (\vec{v}_2)^{i_2} \cdots (\vec{v}_t)^{i_t} &=&
i_1 \vec{v}_i + i_2 \vec{v}_2 + \cdots + t_t \vec{v}_t.
\end{eqnarray}
Since $\vec{v}_1, \vec{v}_2, \ldots, \vec{v}_t$ are linearly
independent, by expression (\ref{exp-a}) we have
\begin{eqnarray} \label{exp-b}
(\vec{v}_1)^{i_1} (\vec{v}_2)^{i_2} \cdots (\vec{v}_t)^{i_t} =
\vec{\bf 0} &\mbox{iff} & i_1 =i_2 = \cdots = i_t = 0.
\end{eqnarray}
The linear independence of $\vec{v}_1, \vec{v}_2, \ldots, \vec{v}_t$
implies that any non-empty subset of these vectors are also linearly
independent. Similar to expression (\ref{exp-b}), this further
implies that, for any $0\le j_i\le m_i$, $i=1, 2,\ldots, t$,
\begin{eqnarray}\label{exp-c}
(\vec{v}_1)^{i_1} (\vec{v}_2)^{i_2} \cdots (\vec{v}_t)^{i_t} =
(\vec{v}_1)^{j_1} (\vec{v}_2)^{j_2} \cdots (\vec{v}_t)^{j_t}
&\mbox{iff} & i_1 =j_1, i_2=j_2, \ldots, \mbox{ and } i_t = j_t.
\end{eqnarray}
Furthermore, since $p$ is prime and $m_i \in Z_p$, we have
\begin{eqnarray}\label{exp-d}
c(i_1,i_2,\ldots c_t) &= &(^{m_1}_{i_1})(^{m_2}_{i_2}) \cdots
(^{m_t}_{i_t})
(p-1)^{i_1+t_2\cdots+i_t} \pmod{p} \nonumber \\
    &\not=& 0 \pmod{p}
\end{eqnarray}
Combining expressions (\ref{exp-c}) and (\ref{exp-d}), we have
\begin{eqnarray}\label{exp-e}
&& ((p-1)\vec{v}_1 + \vec{v}_0)^m_i \cdots ((p-1)\vec{v}_t +
\vec{v}_0)^m_t \nonumber \\
&& = 1 \vec{\bf 0} + \left(\sum_{0\le i_j\le m_j,\  0\le j\le t,\
\mbox{\ and\ } i_1+i_2\cdots+i_t>0} c(i_1,i_2,\ldots,i_t) \cdot
((\vec{v}_1)^{i_1} (\vec{v}_2)^{i_2} \cdots (\vec{v}_t)^{i_t})
\right).
\end{eqnarray}
In the above expression (\ref{exp-e}), all the coefficients are
nonzero, and all the $(m_1+1)(m_2+1) \cdot (m_t+1) \le p^k$
vectors are distinct. Hence, expression (\ref{exp-lem4}) is
obtained.
\end{proof}

{\bf Remark.}  Lemma \ref{lem4} guarantees that replacing
variables in a $p$-monomial by linearly independent vectors will
prevent the monomial from being annihilated. Note that the total
number of distinct vectors in expression \ref{exp-lem4} is at most
$p^k$.

Lemmas \ref{lem3} and \ref{lem4} have laid a basis for designing
randomized algorithms to test $p$-monomials. One additional help
will be drawn from randomized polynomial identity testing over a
finite field. We are ready to present the algorithm and show how
to integrate group algebra with polynomial identity testing to aid
our design. To simplify description, we assume, like in Koutis
\cite{koutis08} and Williams \cite{williams09}, that the degree of
$p$-monomials in a polynomial is at least $k$, provided that
such monomials exist. Otherwise, we can simply multiply some new
variables to the given polynomial to satisfy the requirement.

\begin{theorem}\label{thm-1}
Let $p$ be a prime number. Let $F(x_1,x_2,\ldots,x_n)$ be an
$n$-variate polynomial of degree $k$ represented by an arithmetic
circuit $C$ of size $s(n)$. There is a randomized $O^*(p^k)$
time algorithm to test with high
probability whether $F$ has a $p$-monomial of degree $k$ in its
sum-product expansion.
\end{theorem}

\begin{proof}
Let $d = k+\log_p k + 1$, we consider the group algebra
$Z_p[Z^d_p]$. As in Williams \cite{williams09}, we first expand
the circuit $C$ to a new circuit $C'$ as follows. For each
multiplication gate $g_i$, we attach a new gate $g'_i$ that
multiplies the output of $g_i$ with a new variable $y_i$, and feed
the output of $g'_i$ to the gate that reads the output of $g_i$.
Assume that $C$ has $h$  multiplications gates. Then, $C'$ will
have $h$ new multiplications gates corresponding to new variables
$y_1, y_2, \ldots, y_h$. Let $F'(y_1, y_1, \ldots, y_h, x_1, x_2,
\ldots, x_n)$ be he new polynomial represented by $C'$. The
algorithm for testing whether $F$ has a $p$-monomial of degree $k$
is given in the following.

\begin{quote}
Algorithm \mbox{RT-MLM} (\underline{R}andomized
\underline{T}esting of \underline{M}ulti\underline {l}inear
\underline{M}onomials):
\begin{enumerate}
\item Select uniform random vectors $\vec{v}_1,\ldots,\vec{v}_n
\in Z^d_p-\{\vec{\bf 0}\}$. \item Replace each variable $x_i$ with
$(\vec{v}_i + \vec{v}_0)$, $1\le i \le n$.
 \item Use $C'$ to calculate
\begin{eqnarray}\label{exp-thm-rt-mlm}
F'(y_1,\ldots,y_h,(\vec{v}_1+\vec{v}_0),\ldots,(\vec{v}_n+\vec{v}_0))
& = & \sum_{j=1}^{2^d} f_j(y_1,\ldots,y_h) \cdot \vec{z}_j,
\end{eqnarray}
where each $f_j$ is a polynomial of degree $k$ over the finite
field $Z_p$, and $\vec{z}_j$ with $1\le j\le 2^d$ are the $2^d$ distinct vectors in
$Z^d_p$. \item Perform polynomial identity testing with the
Agrawal and Biswas algorithm \cite{agrawal-biswas03} for every
$f_{j}$ over $Z_p$. Return {\em "yes"} if one of them is not
identical to zero, or {\em "no"} otherwise.
\end{enumerate}
\end{quote}

It follows from Lemma \ref{lem3} that all monomials that are not $p$-monomials in
$F$ (and hence in $F'$) will become zero, when variables $x_i$ is
replaced by $(\vec{v}_i + \vec{v}_0)$ at Step ii. We shall
estimate that with high probability some $p$-monomials will
survive from those replacements, i.e., will not become the zero
element $\bf{0}$ in $Z_p[Z^d_p]$.

Consider any given $p$-monomial $\pi = x_{i_1}^{m_1}\cdots
x_{i_t}^{m_t}$ of degree $k$ with  $1\le m_i < p$ and $k= m_1
+\cdots+m_t$, $i=1,\ldots,t$. For any $1\le j \le t$,
\begin{eqnarray}
\mbox{Pr}\left[\vec{v}_j \in \mbox{span}(\vec{v}_{i_1},\ldots,
\vec{v}_{i_{j-1}})\right] & = & \frac{p^{j-1}}{p^d}, \nonumber
\end{eqnarray}
since $|\mbox{span}(\vec{v}_{i_1},\ldots, \vec{v}_{i_{j-1}})| =
p^{j-1}$ and $|Z_p^d| = p^d$. Hence,
\begin{eqnarray}\label{prop}
& & \mbox{Pr}\left[(\exists j \in \{1,\ldots, t\}) [ \vec{v}_{i_j} \in
\mbox{span}(\vec{v}_{i_1},\ldots, \vec{v}_{i_{j-1}})]\right] \nonumber \\
& & = \mbox{Pr}\left[ [\vec{v}_1 = \vec{\bf 0}] \vee [ \vec{v}_{i_2}
\in \mbox{span}(\vec{v}_{i_1})] \vee
\cdots \vee [ \vec{v}_{i_t} \in \mbox{span}(\vec{v}_{i_1},\ldots,
\vec{v}_{i_{t-1}})] \right] \nonumber \\
& & \le \mbox{Pr}[ \vec{v}_1 = \vec{\bf 0}] +
\mbox{Pr}[\vec{v}_{i_2} \in \mbox{span}(\vec{v}_{i_1})]
+ \cdots + \mbox{Pr}[\vec{v}_{i_t} \in
\mbox{span}(\vec{v}_{i_1},\ldots, \vec{v}_{i_{t-1}})] \nonumber \\
&& = \frac{p^0}{p^d} + \frac{p^1}{p^d} + \cdots +
\frac{p^{t-1}}{p^d} \le  t \frac{p^{t-1}}{p^d} \nonumber \\
& & \le k \frac{p^{k-1}}{p^{k+\log_p k +1}}  \le \frac{1}{p^2} \le
\frac{1}{4}.
\end{eqnarray}
Because $\vec{v}_{i_1},\ldots,\vec{v}_{i_t}$ are linearly
independent iff there is no $\vec{v}_{i_j} \in
\mbox{span}(\vec{v}_{i_1},\ldots, \vec{v}_{i_{j-1}})$, by
expression (\ref{prop}) the probability that
$\vec{v}_{i_1},\ldots,\vec{v}_{i_t}$ are linearly independent is
at least $\frac{3}{4}$. This implies, by Lemma \ref{lem4}, that
the monomial $\pi$ will survive from the replacements at Step ii
with probability at least $\frac{3}{4}$. Furthermore, by
expression (\ref{exp-lem4}) in Lemma \ref{lem4},
\begin{eqnarray}\label{new-poly}
((p-1)\vec{v}_1 + \vec{v}_0)^{m_i} \cdots ((p-1)\vec{v}_t +
\vec{v}_0)^{m_t} &=& \sum^{p^k}_{i=1}c(\pi)_i \vec{u}_i(\pi),
\end{eqnarray}
where $c(\pi)_i$ are coefficients in $Z_p$ such that $(m_1+1)(m_2+1)\cdots (m_t+1)$
 of them are nonzero, and
$\vec{u}_i(\pi)$ are distinct vectors in $Z^d_p$. Let $\psi(\pi)$
be the product of the new variables $y_j$ that are added with
respect to the gates in $C$ such that those gates produce the
monomial $\pi$. Then, $\psi(\pi)$ is a monomial that is generated
by $C'$. Hence, at Step iii, by expression (\ref{new-poly}) $F'$
will have monomials respect to $\pi$ as given in the following expansion:
\begin{eqnarray}\label{new-monomial}
\phi(\pi) & = & \psi(\pi) \cdot ((p-1)\vec{v}_1 + \vec{v}_0)^{m_i}
\cdots ((p-1)\vec{v}_t + \vec{v}_0)^{m_t} \nonumber \\
   &=& \sum^{p^k}_{i=1}c(\pi)_i \cdot \psi(\pi) \cdot  \vec{u}_i(\pi).
\end{eqnarray}
Let ${\cal S}$ be the set of all those $p$-monomials that survive
from the variable replacements. Then,
\begin{eqnarray}\label{new-F}
&&F'(y_1,\ldots,y_h,(\vec{v}_1+\vec{v}_0),\ldots,(\vec{v}_n+\vec{v}_0))
= \sum_{\pi \in {\cal S}} \phi(\pi)  \nonumber \\
& & = \sum_{\pi\in {\cal S}} \left(\sum^{p^k}_{i=1}c(\pi)_i \cdot
\psi(\pi) \cdot \vec{u}_i(\pi)\right) \nonumber \\
& & = \sum_{j=1}^{2^d} \left(\sum_{\pi\in {\cal S} \mbox{ and }
\vec{z}_j = \vec{u}_i(\pi)} c(\pi)_i \cdot \psi(\pi) \right) \cdot
\vec{z}_j
\end{eqnarray}
Let
\begin{eqnarray}
f_j(y_1,\ldots,y_h) & = & \sum_{\pi\in {\cal S} \mbox{ and }
\vec{z}_j = \vec{u}_i(\pi)} c(\pi)_i \cdot \psi(\pi), \nonumber
\end{eqnarray}
then the degree $k$ polynomial with respect to $\vec{z}_j$ is
obtained for $F'$ in expression (\ref{exp-thm-rt-mlm}).

Recall that when constructing the circuit $C'$, each new gate is
associated with a new variable. This means that for any two
monomials $\pi'$ and $\pi''$ in $F$, we have $\psi(\pi')\not=
\psi(\pi'')$. This implies that we cannot add $c(\pi') \cdot
\psi(\pi')$ to $c(\pi'') \cdot \psi(\pi'')$ in $f_j$. Thus, the
possibility of a {\rm "zero-sum"} of coefficients from different
surviving monomials is completely avoided during the construction
of $f_j$. Therefore, conditioned on that ${\cal S}$ is not empty,
$F'$ must not be identical to zero, i.e., there exists at least one
$f_j$ that is not identical to zero. At Step iv, we use the
randomized algorithm by Agrawal and Biswas \cite{agrawal-biswas03}
to test whether $f_j$ is identical to zero. It follows from
Theorem 4.6 in Agrawal and Biswas \cite{agrawal-biswas03} that
this testing can be done with probability at least $\frac{5}{6}$
in time polynomially in $s(n)$ and $\log q$. Since ${\cal S}$ is
not empty with probability at least $\frac{3}{4}$, the probability
of overall success of testing whether $F$ has a $p$-monomial is at
least $\frac{5}{8}$.

Finally, we address the issues about how to calculate $F'$ and the
time needed to do so. Naturally, every element in the group
algebra $Z_p[Z^d_p]$ can be represented by a vector in
$Z^{p^d}_p$. Adding two elements in $Z_p[Z^d_p]$ is equivalent to
adding the two corresponding vectors in $Z_p^{p^d}$, and the
latter can be done in $O(p^d \log p)$ time via component-wise sum.
In addition, multiplying two elements in $Z_p[Z^d_p]$ is
equivalent to multiplying the two corresponding vectors in
$Z_p^{p^d}$, and the latter can be done in $O(dp^d\log^2 p)$ with
the help of a similar Fast Fourier Transform style algorithm as in
Williams \cite{williams09}. Calculating $F'$ consists of $s(n)$
arithmetic operations of either adding or multiplying two elements
in $Z_p[Z^d_p]$ based on the circuit $C$ or $C'$. Hence, the total
time needed is $O(s(n) d p^d log ^2 p)$. At Step iv, we run the
Agrawal and Biswas \cite{agrawal-biswas03} algorithm to $F'$ to
simultaneously testing whether there is one $f_j$ such that $f_j$ is
not identical to zero. We choose a probability $\frac{5}{6}$, the
by Theorem 4.6 in Agrawal and Biswas \cite{agrawal-biswas03}, this
testing can be done in $O^*((s(n))^4 n^4 log^2 p)$ time, suppressing a
$\mbox{poly}(\log s(n), \log n, \log \log p)$ factor.
Recall that $d = k+log_p k +1$. The total time for the entire
algorithm is $O^*(p^k)$.
\end{proof}

\section{Derandomization}

In this section, we turn our attention to formulas instead of
general arithmetic circuits and shall design a deterministic
algorithm to test $p$-monomials for polynomials represented by a
formula. Recall that the algorithm RT-MLM has only two randomized processes
at Step i to select $n$ uniform random variables and
at Step iv to test whether one $f_j$ from $F'$ is identical to
zero over $Z_p$. In this section, we shall derandomize these two
randomized processes respectively with the help of two advanced techniques of
perfect hashing by Chen {\em at al.} \cite{jianer-chen07} and Naor {\em
at el.} \cite{naor95} and noncommunicative multivariate
polynomial identity testing by Raz and Shpilka \cite{raz05}.

Let $n$ and $k$ be two integers such that $1\le k\le n$. Let
${\cal A} =\{1, 2, \ldots, n\}$ and ${\cal K} = \{1, 2, \ldots,
k\}$. A $k$-coloring of the set ${\cal A}$ is a function from
${\cal A}$ to ${\cal K}$. A collection ${\cal F}$ of
$k$-colorings of ${\cal A}$ is a $(n,k)$-family of {\em perfect
hashing functions} if for any subset $W$ of $k$ elements in ${\cal
A}$, there is a $k$-coloring $h \in {\cal F}$ that is injective
from $W$ to ${\cal K}$, i.e., for any $x, y \in W$, $h(x)$ and
$h(y)$ are distinct elements in ${\cal K}$.

\begin{theorem}\label{thm-2}
Let $p$ be a prime number. Let $F(x_1,x_2,\ldots,x_n)$ be an
$n$-variate polynomial of degree $k$ represented by a formula $C$
of size $s(n)$. There is a deterministic $O(6.4^k + p^k)$ time
algorithm to test whether $F$
has a $p$-monomial of degree $k$ in its sum-product expansion.
\end{theorem}

\begin{proof}
As in the proof of Theorem \ref{thm-1}, we consider the group
algebra $Z_p[Z^k_p]$. Here, we do not need to expand the dimension
$k$ to $d>k$.  We also construct a new formula $C'$ from $C$ by
adding new variable $y_i$ for each multiplication gate $g_i$ in
the same way as what we did for Theorem \ref{thm-1}. Assume that
$C$ has $h$ many multiplication gates, then $C'$ will have $h$ new
multiplication gates corresponding to new variables $y_1, y_2,
\ldots, y_h$. The algorithm for testing whether $F$ has a
$p$-monomial of degree $k$ is given as follows.

\begin{quote}
Algorithm \mbox{DT-MLM} (\underline{D}eterministic
\underline{T}esting of \underline{M}ulti\underline{l}inear
\underline{M}onomials):
\begin{enumerate}
\item Construct with the algorithm by Chen {\em at el.}
\cite{jianer-chen07} an $(n,k)$-family of perfect hashing functions
${\cal H}$ of size $O(6.4^k\log^2 n)$. \item Select $k$ linearly
independent vectors $\vec{v}_1,\ldots,\vec{v}_k \in Z^k_p$. (No
randomization is needed at this step.) \item For each perfect
hashing function $\tau \in{\cal H}$ do
\begin{quote}
a. For each variable $x_i$, replace it by $(\vec{v}_{\tau(i)} +
\vec{v}_0)$.

b. Use $C'$ to calculate
\begin{eqnarray}\label{exp-thm-dt-mlm}
&
&F'(y_1,\ldots,y_h,(\vec{v}_1+\vec{v}_0),\ldots,(\vec{v}_n+\vec{v}_0))
\nonumber \\
&& =  \sum_{j=1}^{2^k} f_j(y_1,y_2,\ldots,y_h) \cdot \vec{z}_j,
\end{eqnarray}
where each $f_j$ is a polynomial of degree $k$ over the finite
field $Z_p$, and vectors $\vec{z}_j$ with $1\le j\le 2^k$ are the $2^k$ distinct vectors in
$Z^k_p$.

c. Perform polynomial identity testing with the Raz and Shpilka
algorithm \cite{raz05} for every $f_j$ over $Z_p$. Stop
and return {\em "yes"} if one of them is not identical to zero.
\end{quote}
\item[iv.] If all perfect hashing functions in ${\cal H}$ have
been tried without returning {\em "yes"}, then stop and output
{\em "no"}.
\end{enumerate}
\end{quote}

By Chen {\em at el.}\cite{jianer-chen07}, Step i can be done in $O(6.4^k n \log^2 n)$
times. Step ii can be easily done in $O(k^2\log p)$ time.

It follows from Lemma \ref{lem3} that all those monomials that are not $p$-monomials in
$F$, and hence in $F'$, will be annihilated, when variables $x_i$
are replaced by $(\vec{v}_i + \vec{v}_0)$ at Step iii.a.

Consider any given $p$-monomial $\pi = x_{i_1}^{m_1}\cdots
x_{i_t}^{m_t}$ of degree $k$ with  $1\le m_i < p$ and $k= m_1
+\cdots+m_t$, $i=1,\ldots,t$. Because of the nature of ${\cal H}$,
there is at least one perfect hashing function $\tau$ in ${\cal
H}$ such that $\tau(i_{j'}) \not= \tau(i_{j''})$ if $i_{j'} \not=
i_{j''}$, $1\le j', j'' \le t \le k$. This means that
$\vec{v}_{\tau(i_1)}, \ldots, \vec{v}_{\tau(i_t)}$ are distinct
and hence linearly independent. By Lemma \ref{lem4}, $\pi$ will
survive from the replacements at Step iii.a. Let ${\cal S}$ be the
set of all surviving $p$-monomials. Following the same analysis as
in the proof of Theorem \ref{thm-1}, we have $F'$ that is not
identical to zero if ${\cal S}$ is not empty. That is, there is at
least one $f_j$ that is not identical to zero, if ${\cal S}$ is
not empty. Moreover, the time needed for calculating $F'$ is
$O(kp^k\log^2 p)$.

We now consider imposing noncommunicativity on $C'$ as follows.
Inputs to an arithmetic gate are ordered so that the formal
expressions $y_{i_1}\cdot y_{i_2} \cdot \cdots \cdot y_{i_r}$ and
$y_{j_1}\cdot y_{j_2} \cdot \cdots \cdot y_{j_l}$ are the same iff
$r=l$ and $i_q =j_q$ for $q=1,\ldots,r$. Finally, we use the
algorithm by Raz and Shpilka \cite{raz05} to test whether
$f_j(y_1,\ldots,y_h)$ is identical to zero of not. This can be
done in time polynomially in $s(n)$ and $n$, since $f_j$ is a
non-communicative polynomial represented by a formula.

Combining the above analysis, the total time of the algorithm
DT-MLM is $O(6.4^k n \log^2 n + k p^k (s(n) n)^{O(1)}\log^2 p) = O^*(6.4^k + p^k)$.
\end{proof}

\section{$\Pi_m\Sigma_2\Pi_t \times \Pi_k\Sigma_3$ Polynomials}

It has been proved by Chen and Fu \cite{chen-fu10} that the
problem of testing monomials in $\Pi_m\Sigma_s$ polynomials is
solvable in $(ms\sqrt{m+s})$ time, and in $\Pi_m\Sigma_2\Pi_t$
polynomials is in $O((mt)^2)$ time. On the other hand, it has also been
proved by in \cite{chen-fu10} that the problem for
$\Pi_m\Sigma_3$ and $\Pi_m\Sigma_2\Pi_t \times \Pi_k\Sigma_3$
polynomials is respectively NP-complete. Moreover, a $O(tm^2
1.7751^m)$ time algorithm was obtained for $\Pi_m\Sigma_3\Pi_t$
polynomials, and so was a $O((mt)^2 3^k)$ algorithm obtained for
$\Pi_m\Sigma_2\Pi_t \times \Pi_k\Sigma_3$ polynomials.  In this
section, we shall devise two parameterized algorithms, one
deterministic and the other randomized, for testing multilinear
monomials in $\Pi_m\Sigma_2\Pi_t \times \Pi_k\Sigma_3$
polynomials, improving the $O((mt)^2 3^k)$ upper bound in
\cite{chen-fu10}.

\begin{theorem}\label{yang-thm-1}
There is a deterministic algorithm of time $O(((mt+k)^2+k) 2^k)$
to test whether any $\Pi_m\Sigma_2\Pi_t\times\Pi_k\Sigma_3$
polynomial has a multilinear monomial in its sum-product
expansion.
\end{theorem}

\begin{proof}
Let $F = F_1 \cdot F_2$ such that $F_1 = f_1\cdots f_m$ is a
$\Pi_m\Sigma_2\Pi_t$ polynomial and $F_2=g_1\cdots g_k$ is a
$\Pi_k\Sigma_3$ polynomial, where $f_i = (T_{i1}+T_{i2})$ and
$g_j=(x_{j1}+x_{j2}+x_{j3})$, $1\le i\le m$, $1\le j\le k$.

Consider variable $x_{11}$ in the clause $g_1$. We devise a branch
and bound process to divide  the testing for $F$ into the testing
for two new polynomials.   We eliminate all $x_{11}$ in $g_j$ for
$j=1, \ldots, k$. Let $g'_j$ be the clause resulted from  $g_j$
after the eliminating process. Let $h_1 = F_1 \cdot g'_1$, $h_2 = F_1
\cdot x_{11}$, $q = g'_2 \ldots g'_k$. Note that exactly one of
the three variable $x_{11}, x_{12}$ and $x_{13}$ in the clause
$g_1$ must be selected to form a monomial (hence a
multilinear monomial) for $F$  in the sum-product expansion of $F$. We have
two cases concerning the selection of $x_{11}$:

(1) $x_{11}$ can not be selected to help form any multilinear
monomial. In this case, $F$ has a multilinear monomial, iff $h_1
\cdot q$ has a multilinear monomial.

(2) $x_{11}$ can be selected to form a multilinear monomial. Thus,
$F$ has a multilinear monomial, iff $h_2 \cdot q$ has a
multilinear monomial.

In either case, the new polynomial is a product of two polynomials
with the first being a  $\Pi_{m+1}\Sigma_2\Pi_t$ polynomial and
the second a $\Pi_k\Sigma_3$ polynomial. Furthermore, the second
is the common $q$, which has one fewer clause than $F_2$.

Let $T(k)$ denote the time for testing multilinear monomials in
$F$. Notice that the eliminating process for $x_{11}$ takes $O(k)$
time. Then, $T(k)$ is bounded as follows
$$
T(k) \le 2T(k-1) + O(k) \le 2^k(T(0) + O(k)).
$$
$T(0)$ is the time to test multilinear monomials in a
$\Pi_{m+k}\Sigma_2\Pi_t$ polynomial with a size of $O(mt+k)$. By the
algorithm in \cite{chen-fu10} for this type of polynomials, $T(0)
= O((mt+k)^2)$. Therefore, $T(k) = O(((mt+k)^2 + k) 2^k)$.
\end{proof}

We now show that the upper bound in the above theorem can be further improved via
randomization.

\begin{theorem}\label{yang-thm-2}
There is a $O((mt+k)^2 1.5^k))$ time randomized algorithm that
finds a multilinear monomial for any
$\Pi_m\Sigma_2\Pi_t\times\Pi_k\Sigma_3$ polynomial with
probability at least $1-\frac{1}{e}$ if such monomials exist, or
returns {\em "no''} otherwise.
\end{theorem}

\begin{proof}
Like in Theorem \ref{yang-thm-1}, let $F = F_1 \cdot F_2$ such
that $F_1 = f_1\cdots f_m$ is a $\Pi_m\Sigma_2\Pi_t$ polynomial
and $F_2=g_1\cdots g_k$ is a $\Pi_k\Sigma_3$ polynomial with $f_i
= (T_{i1}+T_{i2})$ and $g_j=(x_{j1}+x_{j2}+x_{j3})$.

Assume that $F$ has a multilinear monomial $\pi$. Then, one of the
three variables in $g_j$ must be included in $\pi$, $1\le j \le
k$. We uniformly select two distinct variables $y_{j1}$ and
$y_{j2}$ from $g_j$, then  $g'_j = (y_{j1}+y_{j2})$ contains a
desired variable for $\pi$ with a probability at least $2/3$. Let
$$
F'=F_1\cdot (g'_1 \cdots g'_k),
$$
then $F'$ has a multilinear monomial with a probability at least
$(\frac{2}{3})^k$. On the other hand, if $F$ does not have any
multilinear monomials in its sum-product expansion, then $F'$ must
not have any multilinear monomials. Notice that $F'$ is a
$\Pi_{m+k}\Sigma_2\Pi_t$ polynomial with a size of $O(mt+k)$. By
the algorithm for this type of polynomials by Chen and Fu in
\cite{chen-fu10}, one can find  a multilinear
monomial in $F'$ in time $O((mt+k)^2)$. In other words, the above
randomized process will fail to find a multilinear monomial in $F$
with a probability of at most $1-\left(\frac{2}{3}\right)^k$ if such monomials
exist, or return {\em "no''} otherwise.

Repeat the above randomized process $\left(\frac{3}{2}\right)^k$ many times.
If $F$ has multilinear monomials, then these processes will fail to
find one with a probability of at most
$$
\left[1-\left(\frac{2}{3}\right)^k\right]^{\left(\frac{3}{2}\right)^k} < \frac{1}{e}.
$$
Hence, the processes will find a multilinear monomial in $F$ with a
probability of at least $1-\frac{1}{e}$. If $F$ does not have any
multilinear monomial, then none of these repeated processes will
find one in $F$. The total time of all the repeated processes is
$O((mt+k)^2 1.5^k)$.
\end{proof}

It is justified in \cite{chen-fu10} that the resemblance of
$\Pi\Sigma\Pi$ and $\Pi\Sigma$ polynomials with SAT formulas is
{\em "superficial"}. For example, The multilinear monomial testing
problem for $\Pi_m\Sigma_3\Pi_1$ polynomials is in P, but 3SAT is
NP-complete. As another example to show such superficial
resemblance, one might consider to apply Sch\"{o}ning's algorithm
for 3SAT \cite{schoning02} to the multilinear monomial testing
problem. However, this is problematic. For the 3SAT problem, it is
easy to find an unsatisfied 3-clause. On the other hand, for the
multilinear monomial testing problem, we do not know which term in
which clause leads to a confliction. Therefore, it is difficult to
decide the change of the hamming distance between the current
solution and any target solution. This difficulty constitutes a
major barrier towards applying the Sch\"{o}ning's algorithm to
monomial testing.

\section{W[1]-Hardness}

Although deterministic and randomized parameterized algorithms
have been devised for testing monomials in previous three sections as well as
in \cite{koutis08,williams09,chen-fu10}, yet we shall prove in
this section that testing some special type of monomials in polynomials
represented by arithmetic circuits is  not
fixed-parameter tractable, unless some unlikely collapse occurs in
the fixed parameter complexity theory.

One shall notice that difference between the general
monomial testing and the specific monomial testing. The former asks for the existence of
{\em "any one''} from a set of possibly  many monomials that are needed. The latter asks for
{\em "a specific one''} from the set. For example, there may be $2^n -1$ multilinear monomials in
the sum-product expansion of a  $n$-variate polynomials. Testing for any one from these many monomials is
certainly different from testing for a specific one, say, $x_1x_3x_7x_{11}$.

Downey and Fellows \cite{downey-fellows95} have established a
hierarchy of parameterized complexity, named the W hierarchy, and
proved that the $k$-Clique problem is W[1]-hard.

\begin{definition}
Let $C = \{i_1, i_2, \ldots, i_k\}$ be a set of $k$ positive integers.
A $k$-clique monomial with respect
to $C$ is the multilinear monomial $\prod_{1\le j < \ell \le
k}x_{i_j i_\ell}$ of degree $\frac{k(k-1)}{2}$.
\end{definition}

\begin{theorem}
It is W[1]-hard to test whether any given $n-$variate polynomial
of degree $\frac{k(k-1)}{2}$ represented by an arithmetic circuit
has a $k$-clique monomial in its sum-product expansion.
\end{theorem}

\begin{proof}
We shall reduce the $k$-clique problem to the $k$-clique monomial
testing problem. Let $G=(V,E)$ be an undirected graph and $k$ an
integer parameter. $V=\{v_1, v_2, \ldots, v_m\}$ is the set of
vertices. Each $(i, j) \in E$ represents the edge connecting
vertices $v_i$ and $v_j$. For each edge $(i, j) \in E$, we define a
variable $x_{ij}$. Let $n= |E|$. We construct a polynomial $f$
with $n$ variables.
\begin{eqnarray}
f(G, 1) & = & 1, \nonumber \\
f(G, 2) &=& \sum_{(i, j)\in E} x_{ij}, \nonumber\\
f(G, t+1) & = & \sum^{m}_{i=1} \left(\sum_{(i,j)\in E}
x_{ij}\right)^t \cdot f(G, t) \nonumber
\end{eqnarray}
As followed from the above definition, $f(G, k)$ has $n=|E|$
variables and its degree is $\frac{k(k-1)}{2}$. It is easy to see
that $f(G, k)$ can be computed by an arithmetic circuit.

If $G$ has a $k$-clique $A=\{i_1,i_2\ldots,i_k\}$, then there are
$\frac{k(k-1)}{2}$ edges connecting any two vertices in $A$. By
definition, $f(G, k)$ has a term $(x_{i_1 i_2} + \cdots + x_{i_1
i_k} + \cdots + x_{i_{k-1}i_k})^{k-1} \cdot f(G, k-1).$ So, we can
select $\pi_1 =x_{i_1 i_2} \cdots x_{i_1 i_k}$ from the first
factor of this term. By simple induction, we can select a
$(k-1)$-clique monomial of degree $\frac{(k-1)(k-2)}{2}$ with
respect to $A - \{i_1\}$. Then, $\pi_1 \cdot \pi_2$ is a
$k$-clique monomial with respect to $A$. On the other hand, it
$f(G, k)$ has a $k$-clique monomial with respect to $A$, then by
definition, $A$ is a $k$-clique for $G$.
\end{proof}

\section*{Acknowledgments}

We thank Ioannis Koutis for helping us understand his group
algebra in \cite{koutis08}.  Bin Fu's research is supported by an
NSF CAREER Award, 2009 April 1 to 2014 March 31.

\end{document}